\RequirePackage{amsmath}
\documentclass[cmp]{svjour}
\usepackage{amsfonts}
\usepackage{cite}
\usepackage[colorlinks=true]{hyperref}

\begin{document}
\title{The conditional Entropy Power Inequality for bosonic quantum systems}
\author{Giacomo De Palma\inst{1} \and Dario Trevisan\inst{2}}
\institute{QMATH, Department of Mathematical Sciences, University of Copenhagen, Universitetsparken 5, 2100 Copenhagen, Denmark \and Universit\`a degli Studi di Pisa, I-56126 Pisa, Italy}
\date{}
\maketitle
\begin{abstract}
We prove the conditional Entropy Power Inequality for Gaussian quantum systems.
This fundamental inequality determines the minimum quantum conditional von Neumann entropy of the output of the beam-splitter or of the squeezing among all the input states where the two inputs are conditionally independent given the memory and have given quantum conditional entropies.
We also prove that, for any couple of values of the quantum conditional entropies of the two inputs, the minimum of the quantum conditional entropy of the output given by the conditional Entropy Power Inequality is asymptotically achieved by a suitable sequence of quantum Gaussian input states.
Our proof of the conditional Entropy Power Inequality is based on a new Stam inequality for the quantum conditional Fisher information and on the determination of the universal asymptotic behaviour of the quantum conditional entropy under the heat semigroup evolution.
The beam-splitter and the squeezing are the central elements of quantum optics, and can model the attenuation, the amplification and the noise of electromagnetic signals.
This conditional Entropy Power Inequality will have a strong impact in quantum information and quantum cryptography.
Among its many possible applications there is the proof of a new uncertainty relation for the conditional Wehrl entropy.
\end{abstract}

\section{Introduction}
The Shannon differential entropy \cite{cover2006elements} of a random variable $X$ with values in $\mathbb{R}^k$ and probability density $p_X(\mathbf{x})\mathrm{d}^kx$ is
\begin{equation}
S(X) := -\int_{\mathbb{R}^k}\ln p_X(\mathbf{x})\;\mathrm{d}p_X(\mathbf{x})\;,
\end{equation}
and quantifies the noise or the information contained in $X$.
Let us consider the linear combination
\begin{equation}\label{eq:linc}
C := \sqrt{\eta}\,A + \sqrt{|1-\eta|}\,B\;,\qquad \eta\ge0
\end{equation}
of two independent random variables $A$ and $B$ with values in $\mathbb{R}^k$.
The classical Entropy Power Inequality \cite{dembo1991information,stam1959some,shannon2001mathematical} states that, if $A$ and $B$ have Shannon differential entropy fixed to the values $S(A)$ and $S(B)$, respectively, the Shannon differential entropy of $C$ is minimized when $A$ and $B$ have a Gaussian probability distribution with proportional covariance matrices:
\begin{equation}\label{eq:cEPI}
\exp\frac{2S(C)}{k} \ge \eta\exp\frac{2S(A)}{k} + \left|1-\eta\right|\exp\frac{2S(B)}{k}\;,
\end{equation}
and is a fundamental element of classical information theory \cite{cover2006elements}.

The noncommutative counterpart of probability measures are quantum states, that are linear positive operators on a Hilbert space with unit trace.
The counterpart of the probability measures on $\mathbb{R}^k$ with even $k$ are the quantum states of a Gaussian quantum system with $n=\frac{k}{2}$ modes.
Gaussian quantum systems \cite{holevo2013quantum,holevo2015gaussian} model electromagnetic waves in the quantum regime.
Electromagnetic waves traveling through cables or free space provide the most promising platform for quantum communication and quantum key distribution \cite{weedbrook2012gaussian}.
Gaussian quantum systems then play a key role in quantum communication and quantum cryptography, and provide the model to determine the maximum communication and key distribution rates achievable in principle by quantum communication devices.
The noncommutative counterpart of the linear combination \eqref{eq:linc} is the beam-splitter of transmissivity $0\le\eta\le1$ or the squeezing of parameter $\eta\ge1$.
The beam-splitter and the squeezing are the fundamental elements of quantum optics, and can model the attenuation, the amplification and the noise of electromagnetic signals.

The quantum counterpart of the Shannon differential entropy is the von Neumann entropy of a quantum state \cite{holevo2013quantum,wilde2017quantum}
\begin{equation}
S(\hat{\rho}) := -\mathrm{Tr}\left[\hat{\rho}\ln\hat{\rho}\right]\;.
\end{equation}
In this paper we prove the conditional Entropy Power Inequality for Gaussian quantum systems (\autoref{thm:EPI}).
Let $A$ and $B$ be the $n$-mode Gaussian quantum systems at the input of the beam-splitter of transmissivity $0\le\eta\le1$ or of the squeezing of parameter $\eta\ge1$, and let $C$ be the $n$-mode Gaussian quantum system at the output.
Let us consider a joint quantum input state $\hat{\rho}_{ABM}$ such that $A$ and $B$ are conditionally independent given the memory system $M$.
This condition is expressed with the vanishing of the quantum conditional mutual information:
\begin{equation}
I(A:B|M) := S(A|M) + S(B|M) - S(AB|M) = 0\;,
\end{equation}
where
\begin{equation}
S(X|M):=S(XM)-S(M)
\end{equation}
is the quantum conditional entropy.
The conditional Entropy Power Inequality determines the minimum quantum conditional entropy of the output $S(C|M)$ among all the quantum input states $\hat{\rho}_{ABM}$ as above and with given quantum conditional entropies $S(A|M)$ and $S(B|M)$:
\begin{equation}\label{eq:EPIm}
\exp\frac{S(C|M)}{n} \ge \eta\exp\frac{S(A|M)}{n} + \left|1-\eta\right|\exp\frac{S(B|M)}{n}\;.
\end{equation}
We also prove that, for any couple of values of $S(A|M)$ and $S(B|M)$, the minimum \eqref{eq:EPIm} for $S(C|M)$ is asymptotically achieved by a suitable sequence of quantum Gaussian input states (\autoref{thm:sharp}).

The conditional Entropy Power Inequality \eqref{eq:EPIm} had been conjectured in \cite{koenig2015conditional}.
It is the conditional version of the quantum Entropy Power Inequality \cite{de2015multimode,de2014generalization,konig2014entropy,konig2016corrections} that provides a lower bound to the von Neumann entropy $S(C)$ of the output of the beam-splitter or of the squeezing for all the product input states $\hat{\rho}_{AB}=\hat{\rho}_A\otimes\hat{\rho}_B$ in terms of the entropies of the inputs $S(A)$ and $S(B)$:
\begin{equation}\label{eq:qEPI}
\exp\frac{S(C)}{n} \ge \eta\exp\frac{S(A)}{n} + \left|1-\eta\right|\exp\frac{S(B)}{n}\;.
\end{equation}
Contrarily to the classical Entropy Power Inequality \eqref{eq:cEPI}, the quantum Entropy Power Inequality \eqref{eq:qEPI} is not saturated by quantum Gaussian states with proportional covariance matrices, unless they have the same entropy \cite{de2014generalization}.

In the classical scenario, the conditional Entropy Power Inequality reads
\begin{equation}\label{eq:cCEPI}
\exp\frac{2S(C|M)}{k} \ge \eta\exp\frac{2S(A|M)}{k} + \left|1-\eta\right|\exp\frac{2S(B|M)}{k}\;,
\end{equation}
where $A$ and $B$ are random variables with values in $\mathbb{R}^k$ and are conditionally independent given the random variable $M$, and $C$ is as in \eqref{eq:linc}.
The conditional Entropy Power Inequality \eqref{eq:cCEPI} is an easy consequence of its unconditioned version \eqref{eq:cEPI}, because the conditional entropy $S(X|M)$ coincides with the expectation value with respect to $M$ of the entropy of $X$ given the value of $M$ (see \autoref{app:cCEPI}):
\begin{equation}\label{eq:X|M}
S(X|M) = \int_M S(X|M=m)\,\mathrm{d}p_M(m)\;.
\end{equation}
The classical conditional Entropy Power Inequality \eqref{eq:cCEPI} is saturated by any joint probability measure on $ABM$ such that, conditioning on any value $m$ of $M$, $A$ and $B$ are independent Gaussian random variables with proportional covariance matrices, and the proportionality constant does not depend on $m$.

In the quantum scenario, conditioning on the value of $M$ is not possible in the presence of entanglement between $AB$ and $M$, and the conditional Entropy Power Inequality is not an easy consequence of the unconditioned Entropy Power Inequality.
The saturation conditions are another fundamental difference between the classical and the quantum scenario.
The quantum conditional Entropy Power Inequality can be saturated only asymptotically by a suitable sequence of quantum Gaussian states.
Contrarily to the classical scenario, the correlation of the inputs $A$ and $B$ with the memory $M$ is necessary for the saturation of the inequality.
Indeed, the unconditioned quantum Entropy Power Inequality is not saturated even asymptotically by quantum Gaussian states.

Entropic inequalities are the main tool to prove upper bounds to quantum communication rates \cite{wilde2017quantum,holevo2013quantum} and to prove the security of quantum key distribution schemes \cite{coles2017entropic}.
In these scenarios, a prominent role is played by entropic inequalities in the presence of quantum memory, where the entropies are conditioned on the knowledge of an external observer holding a memory quantum system.
The quantum conditional Entropy Power Inequality proven in this paper will have a profound impact in quantum information theory and quantum cryptography.
The inequality has been fundamental in the proof of a new uncertainty relation for the conditional Wehrl entropy \cite{de2017uncertainty}.
In \autoref{sec:capacity}, we exploit the inequality to prove an upper bound to the entanglement-assisted classical capacity of a non-Gaussian quantum channel.
This implication has been first considered in \cite{koenig2015conditional}, section III.

The proof of the quantum conditional Entropy Power inequality is based on the evolution with the heat semigroup as in \cite{de2015multimode,de2014generalization,konig2014entropy,koenig2015conditional}.
For simplifying the proof, we reformulate the inequality in the equivalent linear version \eqref{eq:EPIl} through the Legendre transform.
The linear inequality \eqref{eq:EPIl} had been proven with a particular choice of $M$ for Gaussian input states in \cite{koenig2015conditional}, Theorem 8.1.
Our proof consists of two parts:
\begin{itemize}
  \item We prove the quantum conditional Stam inequality (\autoref{thm:Stam}), which provides an upper bound to the quantum conditional Fisher information of the output of the beam-splitter or of the squeezing in terms of the quantum conditional Fisher information of the two inputs.
      This inequality implies that the difference between the two sides of the linear inequality \eqref{eq:EPIl} decreases along the evolution with the heat semigroup.
      The linear version \eqref{eq:Stamlin} of the quantum conditional Stam inequality in the particular case $\lambda=\eta$ had been proven in \cite{koenig2015conditional}, eq. (62).
      The proof of \cite{koenig2015conditional}, as well as the proofs of the unconditioned quantum Stam inequality of \cite{de2015multimode,de2014generalization,konig2014entropy}, are affected by regularity issues.
      Indeed, all these papers define the quantum Fisher information as the Hessian of the relative entropy with respect to the displacements, but they do not prove that this Hessian is well defined.
      We introduce a new integral version of the conditional quantum Fisher information, and we define the conditional quantum Fisher information as the limit of its integral version.
      This definition solves all the previous regularity issues.
      Our proof of the quantum conditional Stam inequality is based on an integral version of the quantum de Bruijn identity (\autoref{thm:dBG}), that relates the increase of the quantum conditional entropy generated by the heat semigroup with the integral quantum conditional Fisher information (\autoref{def:Delta}).
  \item We prove that the quantum conditional entropy has an universal scaling independent on the initial state (\autoref{thm:scaling}) in the infinite time limit under the evolution with the heat semigroup.
      This scaling was already known for Gaussian states (\cite{koenig2015conditional}, Lemma 6.1) and implies that the linear inequality \eqref{eq:EPIl} asymptotically becomes an equality.
      Our proof is based on a more general result (\autoref{thm:liminf}), stating that the minimum quantum conditional entropy of the output of any Gaussian quantum channel is asymptotically achieved by the purification of the thermal quantum Gaussian states with infinite temperature.
\end{itemize}
The paper is structured as follows.
In \autoref{sec:GQS} we present Gaussian quantum systems, the beam-splitter and the squeezing.
In \autoref{sec:Fisher} we present the quantum integral conditional Fisher information, and in \autoref{sec:Stam} we prove the quantum conditional Stam inequality.
In \autoref{sec:scaling} we prove the universal asymptotic scaling of the quantum conditional entropy.
In \autoref{sec:EPI} we prove the quantum conditional Entropy Power Inequality, and in \autoref{sec:sharp} we prove that this inequality is optimal.
In \autoref{sec:capacity} we apply the quantum conditional Entropy Power Inequality to prove an upper bound to the entanglement-assisted classical capacity of a non-Gaussian quantum channel.
We conclude in \autoref{sec:concl}.
\autoref{app:cCEPI} and \autoref{app} contain the proof of the classical conditional Entropy Power Inequality and of the auxiliary lemmas, respectively.

\section{Gaussian quantum systems}\label{sec:GQS}
The Hilbert space of a Gaussian quantum system with $n$ modes is the irreducible representation of the canonical commutation relations
\begin{equation}\label{eq:CCRPQ}
\left[\hat{Q}^k,\;\hat{Q}^l\right] = 0\;,\quad\left[\hat{P}^k,\;\hat{P}^l\right] = 0\;,\quad\left[\hat{Q}^k,\;\hat{P}^l\right] = i\,\delta^{kl}\,\hat{\mathbb{I}}\;,\quad k,\,l=1,\,\ldots,\,n\;.
\end{equation}
We define
\begin{equation}
\hat{R}^{2k-1} := \hat{Q}^k\;,\qquad\hat{R}^{2k} := \hat{P}^k\;,\qquad k=1,\,\ldots,\,n\;,
\end{equation}
and \eqref{eq:CCRPQ} becomes
\begin{equation}
\left[\hat{R}^i,\;\hat{R}^{j}\right] = i\,\Delta^{ij}\,\hat{\mathbb{I}}\;,\qquad i,\,j=1,\,\ldots,\,2n\;,
\end{equation}
where
\begin{equation}
\Delta := \bigoplus_{k=1}^n\left(
                             \begin{array}{cc}
                               0 & 1 \\
                               -1 & 0 \\
                             \end{array}
                           \right)
\end{equation}
is the symplectic form.
The Hamiltonian of the system is
\begin{equation}\label{eq:defH}
\hat{H} = \frac{1}{2}\sum_{i=1}^{2n}\left(\hat{R}^i\right)^2 - \frac{n}{2}\,\hat{\mathbb{I}}\;.
\end{equation}

\begin{definition}[displacement operators]
For any $\mathbf{x}\in\mathbb{R}^{2n}$ we define the displacement operator
\begin{equation}
\hat{D}(\mathbf{x}) := \exp\left(i\sum_{i=1}^{2n}x^i\,\Delta^{-1}_{ij}\,\hat{R}^j\right)\;,
\end{equation}
the unitary operator satisfying for any $i=1,\,\ldots,\,2n$
\begin{equation}
{\hat{D}(\mathbf{x})}^\dag\,\hat{R}^i\,\hat{D}(\mathbf{x}) = \hat{R}^i + x^i\,\hat{\mathbb{I}}\;.
\end{equation}
\end{definition}

\begin{definition}[first moments]
The first moments of a quantum state $\hat{\rho}$ are
\begin{equation}
r^i(\hat{\rho}) := \mathrm{Tr}\left[\hat{R}^i\,\hat{\rho}\right]\;,\qquad i=1,\,\ldots,\,2n\;.
\end{equation}
\end{definition}

\begin{definition}[covariance matrix]
The covariance matrix of a quantum state $\hat{\rho}$ with finite first moments is
\begin{equation}
\sigma^{ij}(\hat{\rho}) := \frac{1}{2}\mathrm{Tr}\left[\left\{\hat{R}^i-r^i(\hat{\rho}),\;\hat{R}^j-r^j(\hat{\rho})\right\}\hat{\rho}\right]\;,\qquad i,\,j=1,\,\ldots,\,2n\;,
\end{equation}
where
\begin{equation}
\left\{\hat{X},\;\hat{Y}\right\} := \hat{X}\,\hat{Y} + \hat{Y}\,\hat{X}
\end{equation}
is the anticommutator.
\end{definition}
\begin{definition}[symplectic eigenvalues]
The symplectic eigenvalues of a real positive matrix $\sigma$ are the absolute values of the eigenvalues of $\Delta^{-1}\sigma$.
\end{definition}

\begin{definition}[heat semigroup]\label{def:heat}
The heat semigroup is the time evolution generated by the convex combination of displacement operators with Gaussian distribution and covariance matrix $t\,I_{2n}$: for any quantum state $\hat{\rho}$
\begin{equation}
\mathcal{N}(t)(\hat{\rho}) := \int_{\mathbb{R}^{2n}} \hat{D}(\mathbf{x})\,\hat{\rho}\,{\hat{D}(\mathbf{x})}^\dag\,\mathrm{e}^{-\frac{|\mathbf{x}|^2}{2t}}\frac{\mathrm{d}^{2n}x}{(2\pi\,t)^n}\;.
\end{equation}
For any $s,\,t\ge0$
\begin{equation}
\mathcal{N}(s)\circ\mathcal{N}(t) = \mathcal{N}(s+t)\;.
\end{equation}
\end{definition}

\subsection{Quantum Gaussian states}
A quantum Gaussian state is a density operator proportional to the exponential of a quadratic polynomial in the quadratures:
\begin{equation}
\hat{\gamma} = \frac{\exp\left(-\frac{1}{2}\sum_{i,\,j=1}^{2n}\left(\hat{R}^i-r^i\right)h_{ij}\left(\hat{R}^j-r^j\right)\right)} {\mathrm{Tr}\exp\left(-\frac{1}{2}\sum_{i,\,j=1}^{2n}\left(\hat{R}^i-r^i\right)h_{ij}\left(\hat{R}^j-r^j\right)\right)}\;,
\end{equation}
where $h$ is a positive real $2n\times2n$ matrix and $\mathbf{r}\in\mathbb{R}^{2n}$.
A thermal Gaussian state is a Gaussian state with zero first moments ($\mathbf{r}=0$) and where the matrix $h$ is proportional to the identity:
\begin{equation}
\hat{\omega} = \frac{\mathrm{e}^{-\beta\hat{H}}}{\mathrm{Tr}\,\mathrm{e}^{-\beta\hat{H}}}\;,\qquad h=\beta\,I_{2n}\;,\qquad \beta>0\;.
\end{equation}
The von Neumann entropy of a quantum Gaussian state is
\begin{equation}
S = \sum_{k=1}^n g\left(\nu_k-\frac{1}{2}\right)\;,
\end{equation}
where
\begin{equation}
g(x):=\left(x+1\right)\ln\left(x+1\right) - x\ln x\;,
\end{equation}
and $\nu_1,\,\ldots,\,\nu_n$ are the symplectic eigenvalues of its covariance matrix.

\subsection{Beam-splitter and squeezing}
Given the $n$-mode Gaussian quantum systems $A$, $B$, $C$ and $D$, the beam-splitter with inputs $A$ and $B$, outputs $C$ and $D$ and transmissivity $0\le\eta\le1$ is implemented by the mixing unitary operator $\hat{U}_\eta:AB\to CD$ acting on the quadratures as \cite{ferraro2005gaussian}
\begin{subequations}
\begin{align}
\hat{U}_\eta^\dag\,\hat{R}_C^i\,\hat{U}_\eta &= \sqrt{\eta}\,\hat{R}_A^i + \sqrt{1-\eta}\,\hat{R}_B^i\;,\\
\hat{U}_\eta^\dag\,\hat{R}_D^i\,\hat{U}_\eta &= -\sqrt{1-\eta}\,\hat{R}_A^i + \sqrt{\eta}\,\hat{R}_B^i\;,\qquad i=1,\,\ldots,\,2n\;.
\end{align}
\end{subequations}
The beam-splitter is a passive element, and does not require energy for functioning.
Indeed, the mixing unitary operator preserves the Hamiltonian \eqref{eq:defH}:
\begin{equation}
\hat{U}_\eta\left(\hat{H}_A + \hat{H}_B\right)\hat{U}_\eta^\dag = \hat{H}_C + \hat{H}_D\;.
\end{equation}

The squeezing unitary operator with parameter $\eta\ge1$ acts on the quadratures as \cite{ferraro2005gaussian}
\begin{subequations}
\begin{align}
\hat{U}_\eta^\dag\,\hat{Q}_C^k\,\hat{U}_\eta &= \sqrt{\eta}\,\hat{Q}_A^k + \sqrt{\eta-1}\,\hat{Q}_B^k\;,\\
\hat{U}_\eta^\dag\,\hat{P}_C^k\,\hat{U}_\eta &= \sqrt{\eta}\,\hat{P}_A^k - \sqrt{\eta-1}\,\hat{P}_B^k\;,\\
\hat{U}_\eta^\dag\,\hat{Q}_D^k\,\hat{U}_\eta &= \sqrt{\eta-1}\,\hat{Q}_A^k + \sqrt{\eta}\,\hat{Q}_B^k\;,\\
\hat{U}_\eta^\dag\,\hat{P}_D^k\,\hat{U}_\eta &= -\sqrt{\eta-1}\,\hat{P}_A^k + \sqrt{\eta}\,\hat{P}_B^k\;,\qquad k=1,\,\ldots,\,n\;.
\end{align}
\end{subequations}
The squeezing acts differently on the $Q^k$ and on the $P^k$.
Indeed, the squeezing is an active operation that requires energy, and the squeezing unitary operator does not preserve the Hamiltonian \eqref{eq:defH}.

We define for any joint quantum state $\hat{\rho}_{AB}$ on $AB$ and any $\eta\ge0$
\begin{equation}
\mathcal{B}_\eta(\hat{\rho}_{AB}) := \mathrm{Tr}_D\left[\hat{U}_\eta\,\hat{\rho}_{AB}\,\hat{U}_\eta^\dag\right]\;.
\end{equation}
$\mathcal{B}_\eta$ implements the beam-splitter for $0\le\eta\le1$ and the squeezing for $\eta\ge1$.

\begin{lemma}[compatibility with displacements]\label{lem:BD}
We have for any quantum state $\hat{\rho}_{AB}$ on $AB$ and any $\mathbf{x},\,\mathbf{y}\in\mathbb{R}^{2n}$
\begin{align}
&\mathcal{B}_\eta\left(\hat{D}_A(\mathbf{x})\,\hat{D}_B(\mathbf{y})\,\hat{\rho}_{AB}\,{\hat{D}_A(\mathbf{x})}^\dag\,{\hat{D}_B(\mathbf{y})}^\dag\right)\nonumber\\
&= \hat{D}_C\left(\sqrt{\eta}\,\mathbf{x} + \sqrt{1-\eta}\,\mathbf{y}\right)\,\mathcal{B}_\eta(\hat{\rho}_{AB})\,{\hat{D}_C\left(\sqrt{\eta}\,\mathbf{x} + \sqrt{1-\eta}\,\mathbf{y}\right)}^\dag
\end{align}
for the beam-splitter with transmissivity $0\le\eta\le1$, and
\begin{align}
&\mathcal{B}_\eta\left(\hat{D}_A(\mathbf{x})\,\hat{D}_B(\mathbf{y})\,\hat{\rho}_{AB}\,{\hat{D}_A(\mathbf{x})}^\dag\,{\hat{D}_B(\mathbf{y})}^\dag\right)\nonumber\\
&= \hat{D}_C\left(\sqrt{\eta}\,\mathbf{x} + \sqrt{\eta-1}\,T\mathbf{y}\right)\,\mathcal{B}_\eta(\hat{\rho}_{AB})\,{\hat{D}_C\left(\sqrt{\eta}\,\mathbf{x} + \sqrt{\eta-1}\,T\mathbf{y}\right)}^\dag
\end{align}
for the squeezing with parameter $\eta\ge1$, where
\begin{equation}\label{eq:defT}
T := \bigoplus_{k=1}^n\left(
                        \begin{array}{cc}
                          1 & 0 \\
                          0 & -1 \\
                        \end{array}
                      \right)
\end{equation}
is the time-reversal matrix that leaves the $Q^k$ unchanged and reverses the sign of the $P^k$.
\end{lemma}

\begin{lemma}[compatibility with heat semigroup]\label{lem:BN}
For any $s,\,t\ge0$
\begin{equation}
\mathcal{B}_\eta\circ \left(\mathcal{N}_A(s)\otimes\mathcal{N}_B(t)\right) = \mathcal{N}_C(\eta s + |1-\eta|t)\circ\mathcal{B}_\eta\;.
\end{equation}
\end{lemma}
\begin{proof}
Follows from \autoref{lem:BD}.
\end{proof}

\section{Quantum integral conditional Fisher information}\label{sec:Fisher}
In this Section, we define the quantum integral conditional Fisher information that will permit us to prove the regularity of the quantum Fisher information of \cite{konig2014entropy,de2015multimode,de2014generalization,koenig2015conditional}.

\begin{definition}[quantum integral conditional Fisher information]\label{def:Delta}
Let $A$ be a Gaussian quantum system with $n$ modes, and $M$ a quantum system.
Let $\hat{\rho}_{AM}$ be a quantum state on $AM$.
For any $t\ge0$, we define the integral Fisher information of $A$ conditioned on $M$ as
\begin{subequations}
\begin{align}
\Delta_{A|M}(\hat{\rho}_{AM})(t) &:= I(A:X|M)_{\hat{\sigma}_{AMX}(t)}\ge0\;,\qquad t>0\;,\\
\Delta_{A|M}(\hat{\rho}_{AM})(0) &:= 0\;,
\end{align}
\end{subequations}
where $X$ is a classical Gaussian random variable with values in $\mathbb{R}^{2n}$ and probability density function
\begin{equation}
\mathrm{d}p_{X}(t)(\mathbf{x}) = \mathrm{e}^{-\frac{|\mathbf{x}|^2}{2t}}\frac{\mathrm{d}^{2n}x}{(2\pi\,t)^n}\;,\qquad \mathbf{x}\in\mathbb{R}^{2n}\;,
\end{equation}
and $\hat{\sigma}_{AMX}(t)$ is the quantum state on $AMX$ such that its marginal on $X$ is $p_{X}(t)$ and for any $\mathbf{x}\in\mathbb{R}^{2n}$
\begin{equation}
\hat{\sigma}_{AM|X=\mathbf{x}}(t) = \hat{D}_A(\mathbf{x})\,\hat{\rho}_{AM}\,{\hat{D}_A(\mathbf{x})}^\dag\;.
\end{equation}
\end{definition}
\begin{remark}
The marginal over $AM$ of $\hat{\sigma}_{AMX}$ is
\begin{equation}
\hat{\sigma}_{AM} = (\mathcal{N}(t)\otimes\mathbb{I}_M)(\hat{\rho}_{AM})\;.
\end{equation}
\end{remark}

The fundamental property of the quantum integral conditional Fisher information is the relation with the increase in the quantum conditional entropy generated by the heat semigroup.

\begin{theorem}[quantum integral conditional de Bruijn identity]\label{thm:dBG}
The quantum integral conditional Fisher information coincides with the increase of the quantum conditional entropy generated by the heat semigroup: for any $t\ge0$,
\begin{equation}
\Delta_{A|M}(\hat{\rho}_{AM})(t) = S(A|M)_{(\mathcal{N}_A(t)\otimes\mathbb{I}_M)(\hat{\rho}_{AM})} - S(A|M)_{\hat{\rho}_{AM}}\;.
\end{equation}
\end{theorem}
\begin{proof}
\begin{align}\label{eq:dB}
I(A:X|M)_{\hat{\sigma}_{AMX}} &= S(A|M)_{\hat{\sigma}_{AMX}} - S(A|MX)_{\hat{\sigma}_{AMX}}\nonumber\\
&= S(A|M)_{\hat{\sigma}_{AM}} - \int_{\mathbb{R}^{2n}} S(A|M)_{\hat{\sigma}_{AM|X=\mathbf{x}}}\,\mathrm{d}p_X(\mathbf{x})\nonumber\\
&= S(A|M)_{\hat{\sigma}_{AM}} - \int_{\mathbb{R}^{2n}} S(A|M)_{\hat{\rho}_{AM}}\,\mathrm{d}p_X(\mathbf{x})\nonumber\\
&= S(A|M)_{\hat{\sigma}_{AM}} - S(A|M)_{\hat{\rho}_{AM}}\;.
\end{align}
\end{proof}

The goal of the remainder of this Section is proving that the quantum integral conditional Fisher information is a continuous, increasing and concave function of time (\autoref{thm:conc}).
This result will permit us to prove the regularity of the quantum Fisher information.

\begin{lemma}[continuity of quantum integral conditional Fisher information]\label{lem:cont}
For any quantum state $\hat{\rho}_{AM}$ such that
\begin{equation}
\mathrm{Tr}_A\left[\hat{H}_A\,\hat{\rho}_A\right]=E_0<\infty\;,\qquad S(\hat{\rho}_M)<\infty
\end{equation}
we have
\begin{equation}\label{eq:A|Mt}
\lim_{t\to0}\Delta_{A|M}((\mathcal{N}(t)\otimes\mathbb{I}_M)(\hat{\rho}_{AM})) = \Delta_{A|M}(\hat{\rho}_{AM})\;.
\end{equation}
\end{lemma}
\begin{proof}
From \autoref{thm:dBG}, the claim is equivalent to
\begin{equation}
\lim_{t\to0}S(A|M)_{(\mathcal{N}(t)\otimes\mathbb{I}_M)(\hat{\rho}_{AM})} = S(A|M)_{\hat{\rho}_{AM}}\;.
\end{equation}
We will then proceed along the same lines of the proof of the continuity of the entropy in the set of the quantum states with bounded average energy (\cite{holevo2013quantum}, Lemma 11.8).
We have (see e.g. \cite{de2017wehrl}, Lemma 2)
\begin{equation}
\lim_{t\to0}\left\|(\mathcal{N}(t)\otimes\mathbb{I}_M)(\hat{\rho}_{AM}) - \hat{\rho}_{AM}\right\|_1 = 0\;.
\end{equation}

Since the quantum entropy is lower semicontinuous (\cite{holevo2013quantum}, Theorem 11.6) and
\begin{equation}
S(A|M)_{(\mathcal{N}(t)\otimes\mathbb{I}_M)(\hat{\rho}_{AM})} = S((\mathcal{N}(t)\otimes\mathbb{I}_M)(\hat{\rho}_{AM})) - S(\hat{\rho}_M)\;,
\end{equation}
we have
\begin{equation}
\liminf_{t\to0}S(A|M)_{(\mathcal{N}(t)\otimes\mathbb{I}_M)(\hat{\rho}_{AM})} \ge S(A|M)_{\hat{\rho}_{AM}}\;.
\end{equation}
On the other hand, we have for any $\beta>0$ and any $0\le t<\epsilon$
\begin{align}
S(A|M)_{(\mathcal{N}(t)\otimes\mathbb{I}_M)(\hat{\rho}_{AM})} &= \beta\,\mathrm{Tr}_A\left[\hat{H}_A\,\mathcal{N}(t)(\hat{\rho}_A)\right] + \ln\mathrm{Tr}_A\mathrm{e}^{-\beta\hat{H}_A}\nonumber\\
&\phantom{=} - S\left((\mathcal{N}(t)\otimes\mathbb{I}_M)(\hat{\rho}_{AM})\left\|\frac{\mathrm{e}^{-\beta\hat{H}_A}}{\mathrm{Tr}_A\mathrm{e}^{-\beta\hat{H}_A}}\otimes\hat{\rho}_M\right.\right)\nonumber\\
&= \beta\left(E_0+n\,t\right) + \ln\mathrm{Tr}_A\mathrm{e}^{-\beta\hat{H}_A}\nonumber\\
&\phantom{=} - S\left((\mathcal{N}(t)\otimes\mathbb{I}_M)(\hat{\rho}_{AM})\left\|\frac{\mathrm{e}^{-\beta\hat{H}_A}}{\mathrm{Tr}_A\mathrm{e}^{-\beta\hat{H}_A}}\otimes\hat{\rho}_M\right.\right)\;,
\end{align}
where
\begin{equation}
S(\hat{\rho}\|\hat{\sigma}) = \mathrm{Tr}\left[\hat{\rho}\left(\ln\hat{\rho}-\ln\hat{\sigma}\right)\right]
\end{equation}
is the quantum relative entropy \cite{holevo2013quantum}.
Since the quantum relative entropy is lower semicontinuous (\cite{holevo2013quantum}, Theorem 11.6),
\begin{align}
\limsup_{t\to0}S(A|M)_{(\mathcal{N}(t)\otimes\mathbb{I}_M)(\hat{\rho}_{AM})} &\le \beta\left(E_0+n\,t\right) + \ln\mathrm{Tr}_A\mathrm{e}^{-\beta\hat{H}_A}\nonumber\\
& \phantom{=} - S\left(\hat{\rho}_{AM}\left\|\frac{\mathrm{e}^{-\beta\hat{H}_A}}{\mathrm{Tr}_A\mathrm{e}^{-\beta\hat{H}_A}}\otimes\hat{\rho}_M\right.\right)\nonumber\\
&= S(A|M)_{\hat{\rho}_{AM}} + \beta\left(E_0 + n\,t - \mathrm{Tr}_A\left[\hat{H}_A\,\hat{\rho}_A\right]\right)\;,
\end{align}
and the claim follows taking the limit $\beta\to0$.
\end{proof}

\begin{lemma}\label{lem:DN}
For any $s,\,t\ge0$,
\begin{equation}
\Delta_{A|M}((\mathcal{N}_A(s)\otimes\mathbb{I}_M)(\hat{\rho}_{AM}))(t) = I(A:X|M)_{(\mathcal{N}_A(s)\otimes\mathbb{I}_M)(\hat{\sigma}_{AMX}(t))}\;,
\end{equation}
where $\hat{\sigma}_{AMX}(t)$ is as in \autoref{def:Delta}.
\end{lemma}
\begin{proof}
We have
\begin{equation}
\Delta_{A|M}((\mathcal{N}_A(s)\otimes\mathbb{I}_M)(\hat{\rho}_{AM}))(t) = I(A:X|M)_{\hat{\tau}_{AMX}(s,t)}\;,
\end{equation}
where $X$ is as in \autoref{def:Delta}, and $\hat{\tau}_{AMX}(s,t)$ is the quantum state on $AMX$ such that its marginal on $X$ is $p_X(t)$, and for any $\mathbf{x}\in\mathbb{R}^{2n}$
\begin{align}
\hat{\tau}_{AM|X=\mathbf{x}}(s,t) &= \hat{D}_A(\mathbf{x})\,(\mathcal{N}_A(s)\otimes\mathbb{I}_M)(\hat{\rho}_{AM})\,{\hat{D}_A(\mathbf{x})}^\dag\nonumber\\
&= (\mathcal{N}_A(s)\otimes\mathbb{I}_M)\left(\hat{D}_A(\mathbf{x})\,\hat{\rho}_{AM}\,{\hat{D}_A(\mathbf{x})}^\dag\right)\nonumber\\
&= (\mathcal{N}_A(s)\otimes\mathbb{I}_M)(\hat{\sigma}_{AM|X=\mathbf{x}}(t))\;.
\end{align}
Hence for any $t\ge0$
\begin{equation}
\hat{\tau}_{AMX}(s,t) = (\mathcal{N}_A(s)\otimes\mathbb{I}_M)(\hat{\sigma}_{AMX}(t))\;,
\end{equation}
and the claim follows.
\end{proof}

\begin{lemma}\label{lem:DPID}
For any $s,\,t\ge0$
\begin{equation}
\Delta_{A|M}((\mathcal{N}_A(s)\otimes\mathbb{I}_M)(\hat{\rho}_{AM}))(t) \le \Delta_{A|M}(\hat{\rho}_{AM})(t)\;.
\end{equation}
\end{lemma}
\begin{proof}
From \autoref{lem:DN}, the claim is equivalent to
\begin{equation}
I(A:X|M)_{(\mathcal{N}_A(s)\otimes\mathbb{I}_M)(\hat{\sigma}_{AMX}(t))} \le I(A:X|M)_{\hat{\sigma}_{AMX}(t)}\;,
\end{equation}
that follows from the data-processing inequality for the quantum mutual information.
\end{proof}

\begin{lemma}\label{lem:split}
For any $s,\,t\ge0$
\begin{align}
\Delta_{A|M}(\hat{\rho}_{AM})(s+t) &= \Delta_{A|M}(\hat{\rho}_{AM})(s) + \Delta_{A|M}((\mathcal{N}_A(s)\otimes\mathbb{I}_M)(\hat{\rho}_{AM}))(t)\nonumber\\
&\ge \Delta_{A|M}(\hat{\rho}_{AM})(s)\;.
\end{align}
\end{lemma}
\begin{proof}
Follows from \autoref{thm:dBG}.
\end{proof}

\begin{theorem}[regularity of quantum integral conditional Fisher information]\label{thm:conc}
For any quantum state $\hat{\rho}_{AM}$ on $AM$ such that
\begin{equation}
\mathrm{Tr}_A\left[\hat{H}_A\,\hat{\rho}_A\right]<\infty\;,\qquad S(\hat{\rho}_M)<\infty\;,
\end{equation}
the quantum integral conditional Fisher information $S(A|M)_{(\mathcal{N}(t)\otimes\mathbb{I}_M)(\hat{\rho}_{AM})}$ is a continuous, increasing and concave function of time.
\end{theorem}
\begin{proof}
The continuity follows from \autoref{lem:cont} and \autoref{lem:split}.
From \autoref{lem:split}, $S(A|M)_{(\mathcal{N}(t)\otimes\mathbb{I}_M)(\hat{\rho}_{AM})}$ is increasing.
We then have to prove that for any $s,\,t\ge0$
\begin{equation}\label{eq:claimconv}
\Delta_{A|M}(\hat{\rho}_{AM})\left(\frac{s+t}{2}\right) \overset{?}{\ge} \frac{\Delta_{A|M}(\hat{\rho}_{AM})(s) + \Delta_{A|M}(\hat{\rho}_{AM})(t)}{2}\;.
\end{equation}
Without lost of generality we can assume $s\le t$.
We can rephrase \eqref{eq:claimconv} as
\begin{align}
&\Delta_{A|M}(\hat{\rho}_{AM})\left(\frac{s+t}{2}\right) - \Delta_{A|M}(\hat{\rho}_{AM})(s)\nonumber\\
&\overset{?}{\ge} \Delta_{A|M}(\hat{\rho}_{AM})(t) - \Delta_{A|M}(\hat{\rho}_{AM})\left(\frac{s+t}{2}\right)\;,
\end{align}
that thanks to \autoref{lem:split} is equivalent to
\begin{align}\label{eq:diff}
&\Delta_{A|M}(\hat{\rho}_{AM}(s))\left(\frac{t-s}{2}\right)\nonumber\\
&\overset{?}{\ge} \Delta_{A|M}\left(\left(\mathcal{N}_A\left(\frac{t-s}{2}\right)\otimes\mathbb{I}_M\right)(\hat{\rho}_{AM}(s))\right)\left(\frac{t-s}{2}\right)\;,
\end{align}
where
\begin{equation}
\hat{\rho}_{AM}(s) := (\mathcal{N}_A(s)\otimes\mathbb{I}_M)(\hat{\rho}_{AM})\;.
\end{equation}
Finally, \eqref{eq:diff} holds from \autoref{lem:DPID}.
\end{proof}

\section{Quantum conditional Fisher information and quantum Stam inequality}\label{sec:Stam}
In this Section, we derive the quantum conditional Fisher information and the quantum conditional de Bruijn identity from their integral versions presented in \autoref{sec:Fisher}, and we prove that the quantum conditional Fisher information satisfies the quantum Stam inequality.

\begin{definition}[quantum conditional Fisher information]\label{defn:J}
Let $\hat{\rho}_{AM}$ be a quantum state on $AM$ satisfying the hypotheses of \autoref{thm:conc}.
The Fisher information of $A$ conditioned on $M$ is
\begin{equation}\label{eq:defJ}
J(A|M)_{\hat{\rho}_{AM}} := \lim_{t\to 0}\frac{\Delta_{A|M}(\hat{\rho}_{AM})(t)}{t}\;.
\end{equation}
\end{definition}
\begin{remark}
Since from \autoref{thm:conc} the function $t\mapsto \Delta_{A|M}(\hat{\rho}_{AM})(t)$ is continuous and concave, the limit in \eqref{eq:defJ} always exists (finite or infinite).
\end{remark}
\begin{remark}
\autoref{defn:J} is equivalent to the definition of \cite{koenig2015conditional}, eq. (53).
\end{remark}

\begin{proposition}[quantum conditional de Bruijn identity]\label{prop:ddB}
The quantum conditional Fisher information coincides with the time derivative of the quantum conditional entropy under the heat semigroup evolution:
\begin{equation}
J(A|M)_{\hat{\rho}_{AM}} = \left.\frac{\mathrm{d}}{\mathrm{d}t}S(A|M)_{(\mathcal{N}_A(t)\otimes\mathbb{I}_M)(\hat{\rho}_{AM})}\right|_{t=0}\;.
\end{equation}
\end{proposition}
\begin{proof}
Follows from \autoref{thm:dBG}.
\end{proof}
\begin{remark}
\autoref{prop:ddB} had been proven in \cite{koenig2015conditional}, Theorem 7.3 in the particular case where $\hat{\rho}_{AM}$ is a quantum Gaussian state.
\end{remark}

\begin{theorem}[quantum conditional Stam inequality]\label{thm:Stam}
Let $A$, $B$ and $C$ be Gaussian quantum systems with $n$ modes, $M$ a quantum system, and $\mathcal{B}_\eta:AB\to C$ the beam-splitter with transmissivity $0\le\eta\le1$ or the squeezing with parameter $\eta\ge1$.
Let $\hat{\rho}_{ABM}$ be a quantum state on $ABM$ such that
\begin{equation}
\mathrm{Tr}_{AB}\left[\left(\hat{H}_A+\hat{H}_B\right)\hat{\rho}_{AB}\right]<\infty\;,\qquad S(\hat{\rho}_M)<\infty\;,
\end{equation}
and let us suppose that $A$ and $B$ are conditionally independent given $M$:
\begin{equation}
I(A:B|M)_{\hat{\rho}_{ABM}} = 0\;,
\end{equation}
Then, for any $0\le\lambda\le1$ the quantum linear conditional Stam inequality holds:
\begin{equation}\label{eq:Stamlin}
J(C|M)_{\hat{\rho}_{CM}} \le \frac{\lambda^2}{\eta}J(A|M)_{\hat{\rho}_{AM}} + \frac{(1-\lambda)^2}{|1-\eta|}J(B|M)_{\hat{\rho}_{BM}}\;,
\end{equation}
where
\begin{equation}
\hat{\rho}_{CM} := (\mathcal{B}_\eta\otimes\mathbb{I}_M)(\hat{\rho}_{ABM})\;.
\end{equation}
The quantum conditional Stam inequality follows minimizing over $\lambda$ the right-hand side of \eqref{eq:Stamlin}:
\begin{equation}\label{eq:Stam}
\frac{1}{J(C|M)_{\hat{\rho}_{CM}}} \ge \frac{\eta}{J(A|M)_{\hat{\rho}_{AM}}} + \frac{|1-\eta|}{J(B|M)_{\hat{\rho}_{BM}}}\;.
\end{equation}
\end{theorem}
\begin{remark}
The linear Stam inequality \eqref{eq:Stamlin} had been proven in the particular case $\lambda=\eta$ in \cite{koenig2015conditional}, eq. (62).
\end{remark}
\begin{proof}
We will prove the following inequality for the quantum integral conditional Fisher information:
\begin{equation}\label{eq:StamD}
\Delta_{C|M}(\hat{\rho}_{CM})(t) \le \Delta_{A|M}(\hat{\rho}_{AM})\left(\frac{\lambda^2\,t}{\eta}\right) + \Delta_{B|M}(\hat{\rho}_{BM})\left(\frac{(1-\lambda)^2\,t}{|1-\eta|}\right)\;.
\end{equation}
The quantum linear conditional Stam inequality \eqref{eq:Stamlin} follows taking the derivative of \eqref{eq:StamD} in $t=0$.
The quantum conditional Stam inequality \eqref{eq:Stam} follows choosing
\begin{equation}
\lambda = \frac{\eta\,J(B|M)_{\hat{\rho}_{BM}}}{\eta\,J(B|M)_{\hat{\rho}_{BM}} + |1-\eta|\,J(A|M)_{\hat{\rho}_{AM}}}\;,
\end{equation}
that minimizes the right-hand side of \eqref{eq:Stamlin}.

For any $t\ge0$
\begin{equation}
\Delta_{C|M}(\hat{\rho}_{CM})(t) = I(C:Z|M)_{\hat{\sigma}_{CMZ}(t)}\;,
\end{equation}
where $Z$ is a Gaussian random variable with values in $\mathbb{R}^{2n}$ and probability density function
\begin{equation}
\mathrm{d}p_{Z}(t)(\mathbf{z}) = \mathrm{e}^{-\frac{|\mathbf{z}|^2}{2t}}\frac{\mathrm{d}^{2n}z}{(2\pi\,t)^n}\;,\qquad \mathbf{z}\in\mathbb{R}^{2n}\;,
\end{equation}
and $\hat{\sigma}_{CMZ}(t)$ is the quantum state on $CMZ$ such that its marginal on $Z$ is $p_{Z}(t)$ and for any $\mathbf{z}\in\mathbb{R}^{2n}$
\begin{equation}
\hat{\sigma}_{CM|Z=\mathbf{z}}(t) = \hat{D}_C(\mathbf{z})\,\hat{\rho}_{CM}\,{\hat{D}_C(\mathbf{z})}^\dag\;.
\end{equation}
We define the quantum state $\hat{\sigma}_{ABMZ}(t)$ on $ABMZ$ such that its marginal on $Z$ is $p_{Z}(t)$ and for any $\mathbf{z}\in\mathbb{R}^{2n}$
\begin{equation}
\hat{\sigma}_{ABM|Z=\mathbf{z}} = \hat{D}_A\left(\frac{\lambda\mathbf{z}}{\sqrt{\eta}}\right)\hat{D}_B\left(\frac{(1-\lambda)\mathbf{z}}{\sqrt{1-\eta}}\right)\hat{\rho}_{ABM} {\hat{D}_A\left(\frac{\lambda\mathbf{z}}{\sqrt{\eta}}\right)}^\dag{\hat{D}_B\left(\frac{(1-\lambda)\mathbf{z}}{\sqrt{1-\eta}}\right)}^\dag
\end{equation}
if $0\le\eta\le1$, and
\begin{align}
&\hat{\sigma}_{ABM|Z=\mathbf{z}}\nonumber\\
&= \hat{D}_A\left(\frac{\lambda\mathbf{z}}{\sqrt{\eta}}\right)\hat{D}_B\left(\frac{(1-\lambda)T\mathbf{z}}{\sqrt{\eta-1}}\right)\hat{\rho}_{ABM} {\hat{D}_A\left(\frac{\lambda\mathbf{z}}{\sqrt{\eta}}\right)}^\dag{\hat{D}_B\left(\frac{(1-\lambda)T\mathbf{z}}{\sqrt{\eta-1}}\right)}^\dag
\end{align}
if $\eta\ge1$, where $T$ is the time-reversal matrix defined in \eqref{eq:defT}.
We then have for any $t\ge0$
\begin{equation}
\hat{\sigma}_{CMZ}(t) = (\mathcal{B}_\eta\otimes\mathbb{I}_{MZ})(\hat{\sigma}_{ABMZ}(t))\;.
\end{equation}
We have
\begin{align}
I(A:B|MZ)_{\hat{\sigma}_{ABMZ}} &= \int_{\mathbb{R}^{2n}} I(A:B|M)_{\hat{\sigma}_{ABM|Z=\mathbf{z}}}\,\mathrm{d}p_Z(\mathbf{z})\nonumber\\
&= \int_{\mathbb{R}^{2n}} I(A:B|M)_{\hat{\rho}_{ABM}}\,\mathrm{d}p_Z(\mathbf{z}) = 0\;.
\end{align}
We then have from the data-processing inequality for the quantum mutual information
\begin{align}\label{eq:inI}
I(C:Z|M)_{\hat{\sigma}_{CMZ}} &\le I(AB:Z|M)_{\hat{\sigma}_{ABMZ}}\nonumber\\
&= I(A:Z|M)_{\hat{\sigma}_{AMZ}} + I(B:Z|M)_{\hat{\sigma}_{BMZ}}\nonumber\\
&\phantom{=} + I(A:B|MZ)_{\hat{\sigma}_{ABMZ}} - I(A:B|M)_{\hat{\sigma}_{ABM}}\nonumber\\
&\le I(A:Z|M)_{\hat{\sigma}_{AMZ}} + I(B:Z|M)_{\hat{\sigma}_{BMZ}}\;.
\end{align}
Proceeding as in the proof of \autoref{thm:dBG} we get
\begin{equation}
I(A:Z|M)_{\hat{\sigma}_{AMZ}(t)} = S(A|M)_{\hat{\sigma}_{AM}(t)} - S(A|M)_{\hat{\rho}_{AM}}\;,
\end{equation}
where
\begin{align}
\hat{\sigma}_{AM}(t) &= \int_{\mathbb{R}^{2n}}\hat{D}_A\left(\frac{\lambda\,\mathbf{z}}{\sqrt{\eta}}\right)\, \hat{\rho}_{AM}\,{\hat{D}_A\left(\frac{\lambda\,\mathbf{z}}{\sqrt{\eta}}\right)}^\dag \,\mathrm{e}^{-\frac{|\mathbf{z}|^2}{2t}}\frac{\mathrm{d}^{2n}z}{(2\pi\,t)^n}\nonumber\\
&= \left(\mathcal{N}_A\left(\frac{\lambda^2\,t}{\eta}\right)\otimes\mathbb{I}_M\right)(\hat{\rho}_{AM})\;,
\end{align}
hence
\begin{equation}
I(A:Z|M)_{\hat{\sigma}_{AMZ}(t)} = \Delta_{A|M}(\hat{\rho}_{AM})\left(\frac{\lambda^2\,t}{\eta}\right)\;.
\end{equation}
We can analogously show that
\begin{equation}
I(B:Z|M)_{\hat{\sigma}_{BMZ}(t)} = \Delta_{B|M}(\hat{\rho}_{BM})\left(\frac{(1-\lambda)^2\,t}{|1-\eta|}\right)\;,
\end{equation}
and \eqref{eq:inI} becomes the claim \eqref{eq:StamD}.
\end{proof}

\section{Universal asymptotic scaling of quantum conditional entropy}\label{sec:scaling}
In this Section, we prove that the quantum conditional entropy has a universal asymptotic scaling in the infinite-time limit under the heat semigroup evolution (\autoref{thm:scaling}).
The proof is based on the following more general result, that provides a new universal lower bound for the conditional entropy of the output of any quantum Gaussian channel.
\begin{theorem}[universal lower bound for quantum conditional entropy]\label{thm:liminf}
Let $A$, $B$ be quantum Gaussian systems with $m$ and $n$ modes, respectively, and $\Phi:A\to B$ a quantum Gaussian channel.
Let $\hat{\rho}_{AM}$ be a quantum state on $AM$ such that
\begin{equation}
\mathrm{Tr}_A\left[\hat{H}_A\,\hat{\rho}_A\right]<\infty\;,\qquad S(\hat{\rho}_M)<\infty\;.
\end{equation}
Then, for any quantum system $M$ and any joint quantum state $\hat{\rho}_{AM}$
\begin{equation}
S(B|M)_{(\Phi\otimes\mathbb{I}_M)(\hat{\rho}_{AM})} \ge \lim_{\nu\to\infty}S(B|A')_{(\Phi\otimes\mathbb{I}_{A'})(\hat{\omega}_{AA'}(\nu))}\;,
\end{equation}
where $A'$ is a Gaussian quantum system with $m$ modes, and for any $\nu\ge\frac{1}{2}$, $\hat{\omega}_{AA'}(\nu)$ is a purification of the thermal Gaussian state $\hat{\omega}_A(\nu)$ on $A$ with covariance matrix $\nu\,I_{2m}$.
\end{theorem}
\begin{proof}
Since the quantum conditional entropy is concave, we can restrict to $\hat{\rho}_{AM}$ pure.
Let $K:\mathbb{R}^{2m}\to\mathbb{R}^{2n}$ be the matrix such that for any $\mathbf{x}\in\mathbb{R}^{2n}$
\begin{equation}
(\Phi\otimes\mathbb{I}_M)\left(\hat{D}_A(\mathbf{x})\,\hat{\rho}_{AM}\,{\hat{D}_A(\mathbf{x})}^\dag\right) = \hat{D}_B(K\mathbf{x})\,(\Phi\otimes\mathbb{I}_M)(\hat{\rho}_{AM})\,{\hat{D}_B(K\mathbf{x})}^\dag\;.
\end{equation}
Let $\hat{\rho}_A$ be the marginal of $\hat{\rho}_{AM}$ on $A$.
Since the quantum conditional mutual information is invariant under local unitaries, we can assume that $\hat{\rho}_A$ has zero first moments.
We have
\begin{equation}
S(B|M)_{(\Phi\otimes\mathbb{I}_M)(\hat{\rho}_{AM})} = S\left((\Phi\otimes\mathbb{I}_M)(\hat{\rho}_{AM})\right) - S(\hat{\rho}_M) = S\left(\tilde{\Phi}(\hat{\rho}_A)\right) - S(\hat{\rho}_A)\;,
\end{equation}
where $\hat{\rho}_M$ is the marginal state of $\hat{\rho}_{AM}$, and $\tilde{\Phi}$ is the complementary channel of $\Phi$.
Let $\hat{\gamma}_A$ be the quantum Gaussian state with the same first and second moments as $\hat{\rho}_A$.
Since $\tilde{\Phi}$ is a Gaussian channel, $\tilde{\Phi}(\hat{\gamma}_A)$ is the quantum Gaussian state with the same first and second moments as $\tilde{\Phi}(\hat{\rho}_A)$.
We then have from \autoref{lem:GS}
\begin{align}
&S\left(\tilde{\Phi}(\hat{\rho}_A)\right) - S(\hat{\rho}_A)\nonumber\\
&= S\left(\tilde{\Phi}(\hat{\gamma}_A)\right) - S(\hat{\gamma}_A) + S(\hat{\rho}_A\|\hat{\gamma}_A) - S\left(\left.\tilde{\Phi}(\hat{\rho}_A)\right\|\tilde{\Phi}(\hat{\gamma}_A)\right)\nonumber\\
&\ge S\left(\tilde{\Phi}(\hat{\gamma}_A)\right) - S(\hat{\gamma}_A)\;,
\end{align}
where we have used the data-processing inequality for the quantum relative entropy.
Let $\alpha$ be the covariance matrix of $\hat{\gamma}_A$.
We then have
\begin{equation}
\hat{\omega}_A\left(\|\alpha\|_\infty\right) = \int_{\mathbb{R}^{2m}} \hat{D}_A(\mathbf{x})\,\hat{\gamma}_A\,{\hat{D}_A(\mathbf{x})}^\dag\,\mathrm{d}p_X(\mathbf{x})\;,
\end{equation}
where $p_X$ is the probability distribution of the classical Gaussian random variable $X$ with values in $\mathbb{R}^{2m}$, zero mean and covariance matrix $\|\alpha\|_\infty I_{2m} - \alpha$.
We then have from \autoref{lem:diff}
\begin{equation}
S\left(\tilde{\Phi}(\hat{\gamma}_A)\right) - S(\hat{\gamma}_A) \ge S\left(\tilde{\Phi}\left(\hat{\omega}_A\left(\|\alpha\|_\infty\right)\right)\right) - S\left(\hat{\omega}_A\left(\|\alpha\|_\infty\right)\right)\;.
\end{equation}
\autoref{lem:diff} also implies that the function
\begin{equation}
\nu\mapsto S\left(\tilde{\Phi}\left(\hat{\omega}_A(\nu)\right)\right) - S\left(\hat{\omega}_A(\nu)\right)\;,\qquad \nu\ge\frac{1}{2}
\end{equation}
is decreasing, hence
\begin{align}
&S\left(\tilde{\Phi}\left(\hat{\omega}_A\left(\|\alpha\|_\infty\right)\right)\right) - S\left(\hat{\omega}_A\left(\|\alpha\|_\infty\right)\right) \ge \lim_{\nu\to\infty} \left(S\left(\tilde{\Phi}\left(\hat{\omega}_A(\nu)\right)\right) - S\left(\hat{\omega}_A(\nu)\right)\right)\nonumber\\
&=\lim_{\nu\to\infty} \left(S\left((\Phi\otimes\mathbb{I}_{A'})(\hat{\omega}_{AA'}(\nu))\right) - S\left(\hat{\omega}_{A'}(\nu)\right)\right) = \lim_{\nu\to\infty} S(B|A')_{(\Phi\otimes\mathbb{I}_{A'})(\hat{\omega}_{AA'}(\nu))}\;.
\end{align}
\end{proof}

\begin{lemma}\label{lem:liminf}
For any $t>0$,
\begin{equation}
\lim_{\nu\to\infty}S(A|A')_{(\mathcal{N}(t)\otimes\mathbb{I}_{A'})(\hat{\omega}_{AA'}(\nu))} = n\ln t + n\;.
\end{equation}
\end{lemma}
\begin{proof}
For any $\nu\ge\frac{1}{2}$, the quantum Gaussian state $\hat{\omega}_{AA'}(\nu)$ is the tensor product of $n$ identical two-mode squeezed quantum Gaussian states, each with covariance matrix
\begin{equation}
\alpha(\nu) := \left(
           \begin{array}{cc|cc}
             \nu & 0 & \sqrt{\nu^2-\frac{1}{4}} & 0 \\
             0 & \nu & 0 & -\sqrt{\nu^2-\frac{1}{4}} \\
             \hline
             \sqrt{\nu^2-\frac{1}{4}} & 0 & \nu & 0 \\
             0 & -\sqrt{\nu^2-\frac{1}{4}} & 0 & \nu \\
           \end{array}
         \right)\;,
\end{equation}
where the block decomposition refers to the $AA'$ bipartition.
For any $t\ge0$, the quantum Gaussian state $(\mathcal{N}(t)\otimes\mathbb{I}_{A'})(\hat{\omega}_{AA'}(\nu))$ is the tensor product of $n$ identical two-mode quantum Gaussian states, each with covariance matrix
\begin{equation}
\alpha(\nu,t) := \left(
           \begin{array}{cc|cc}
             \nu+t & 0 & \sqrt{\nu^2-\frac{1}{4}} & 0 \\
             0 & \nu+t & 0 & -\sqrt{\nu^2-\frac{1}{4}} \\
             \hline
             \sqrt{\nu^2-\frac{1}{4}} & 0 & \nu & 0 \\
             0 & -\sqrt{\nu^2-\frac{1}{4}} & 0 & \nu \\
           \end{array}
         \right)\;.
\end{equation}
The symplectic eigenvalues of $\alpha(\nu,t)$ are
\begin{equation}
\nu_\pm(\nu,t) = \frac{1}{2}\sqrt{4\nu\,t\pm 2t\sqrt{4\nu\,t + t^2 + 1} + 2t^2 + 1} = \sqrt{\nu\,t} + \mathcal{O}(1)
\end{equation}
for $\nu\to\infty$, hence
\begin{align}
&\lim_{\nu\to\infty}S(A|A')_{(\mathcal{N}(t)\otimes\mathbb{I}_{A'})(\hat{\omega}_{AA'}(\nu))}\nonumber\\
&= \lim_{\nu\to\infty}n\left(g\left(\nu_+(\nu,t)-\frac{1}{2}\right) + g\left(\nu_-(\nu,t)-\frac{1}{2}\right) - g\left(\nu-\frac{1}{2}\right)\right)\nonumber\\
&= n\ln t + n\;,
\end{align}
where we used that for $\nu\to\infty$
\begin{equation}
g\left(\nu-\frac{1}{2}\right) = \ln\nu + 1 + \mathcal{O}\left(\frac{1}{\nu^2}\right)\;.
\end{equation}
\end{proof}

\begin{theorem}[universal asymptotic scaling of quantum conditional entropy]\label{thm:scaling}
Let $A$ be a Gaussian quantum system with $n$ modes, and $M$ a quantum system.
Let $\hat{\rho}_{AM}$ be a quantum state on $AM$ such that its marginal on $A$ has finite first and second moments, and its marginal on $M$ has finite entropy.
Then,
\begin{equation}\label{eq:scaling}
\lim_{t\to\infty}\left(S(A|M)_{(\mathcal{N}(t)\otimes\mathbb{I}_M)(\hat{\rho}_{AM})} - n\ln t - n\right) = 0\;.
\end{equation}
\end{theorem}
\begin{remark}
The scaling \eqref{eq:scaling} had been proven in \cite{koenig2015conditional}, Lemma 6.1 in the particular case where $\hat{\rho}_{AM}$ is a Gaussian state.
\end{remark}
\begin{proof}
From the subadditivity of the quantum entropy we have for any $t\ge0$
\begin{equation}
S(A|M)_{(\mathcal{N}(t)\otimes\mathbb{I}_M)(\hat{\rho}_{AM})} \le S\left(\mathcal{N}(t)(\hat{\rho}_A)\right)\;.
\end{equation}
Let
\begin{equation}
E:=\frac{1}{n}\mathrm{Tr}_A\left[\hat{H}_A\,\hat{\rho}_A\right]
\end{equation}
be the average energy per mode of $\hat{\rho}_A$, and let $\hat{\omega}_A$ be the thermal quantum Gaussian state with average energy per mode $E$ and covariance matrix $\left(E+\frac{1}{2}\right)I_{2n}$.
For any $t\ge0$, $\mathcal{N}(t)(\hat{\omega}_A)$ is the thermal quantum Gaussian state with the same average energy as $\mathcal{N}(t)(\hat{\rho}_A)$.
We then have from \autoref{lem:Topt}
\begin{equation}
S\left(\mathcal{N}(t)(\hat{\rho}_A)\right) \le S\left(\mathcal{N}(t)(\hat{\omega}_A)\right) = n\,g(E+t)\;,
\end{equation}
hence
\begin{align}
\limsup_{t\to\infty}\left(S(A|M)_{(\mathcal{N}(t)\otimes\mathbb{I}_M)(\hat{\rho}_{AM})} - n\ln t -n\right) &\le \limsup_{t\to\infty}n\left(g(E+t) - \ln t - 1\right)\nonumber\\
&= 0\;.
\end{align}
On the other hand, from \autoref{thm:liminf} and \autoref{lem:liminf} we have for any $t\ge0$
\begin{equation}
S(A|M)_{(\mathcal{N}(t)\otimes\mathbb{I}_M)(\hat{\rho}_{AM})} \ge \lim_{\nu\to\infty}S(A|A')_{(\mathcal{N}(t)\otimes\mathbb{I}_{A'})(\hat{\omega}_{AA'}(\nu))} = n\ln t + n\;.
\end{equation}
\end{proof}

\section{Quantum conditional Entropy Power Inequality}\label{sec:EPI}
In this Section we prove the conditional Entropy Power Inequality, the main result of this paper.

\begin{theorem}[quantum conditional Entropy Power Inequality]\label{thm:EPI}
Let $A$, $B$ and $C$ be Gaussian quantum systems with $n$ modes, $M$ a quantum system, and $\mathcal{B}_\eta:AB\to C$ the beam-splitter with transmissivity $0\le\eta\le1$ or the squeezer of parameter $\eta\ge1$.
Let $\hat{\rho}_{ABM}$ be a quantum state on $ABM$ such that
\begin{equation}
\mathrm{Tr}_{AB}\left[\left(\hat{H}_A+\hat{H}_B\right)\hat{\rho}_{AB}\right]<\infty\;,\qquad S(\hat{\rho}_M)<\infty\;,
\end{equation}
and let us suppose that $A$ and $B$ are conditionally independent given $M$:
\begin{equation}
I(A:B|M)_{\hat{\rho}_{ABM}} = 0\;.
\end{equation}
Then, for any $0\le\lambda\le1$ the quantum linear conditional Entropy Power Inequality holds:
\begin{align}\label{eq:EPIl}
\frac{S(C|M)_{\hat{\rho}_{CM}}}{n} &\ge \lambda\,\frac{S(A|M)_{\hat{\rho}_{AM}}}{n} + \left(1-\lambda\right)\frac{S(B|M)_{\hat{\rho}_{BM}}}{n}\nonumber\\
&\phantom{\ge} + \lambda\ln\frac{\eta}{\lambda} + (1-\lambda)\ln\frac{|1-\eta|}{1-\lambda}\;.
\end{align}
The quantum conditional Entropy Power Inequality follows maximizing over $\lambda$ the right-hand side of \eqref{eq:EPIl}:
\begin{equation}\label{eq:EPI}
\exp\frac{S(C|M)_{\hat{\rho}_{CM}}}{n} \ge \eta\,\exp\frac{S(A|M)_{\hat{\rho}_{AM}}}{n} + \left|1-\eta\right|\exp\frac{S(B|M)_{\hat{\rho}_{BM}}}{n}\;.
\end{equation}
\end{theorem}
\begin{remark}
The quantum linear inequality \eqref{eq:EPIl} had been proven for $\lambda=\eta$ in the special case where $\hat{\rho}_{ABM}$ is a quantum Gaussian state in \cite{koenig2015conditional}, Theorem 8.1.
\end{remark}
\begin{proof}
Let us define for any $t\ge0$
\begin{align}
\hat{\rho}_{ABM}(t) &:= \left(\mathcal{N}_A\left(\frac{\lambda\, t}{\eta}\right)\otimes\mathcal{N}_B\left(\frac{(1-\lambda)\,t}{|1-\eta|}\right)\otimes\mathbb{I}_M\right)(\hat{\rho}_{ABM})\;,\nonumber\\
\hat{\rho}_{CM}(t) &:= (\mathcal{B}_\eta\otimes\mathbb{I}_M)(\hat{\rho}_{ABM}(t))\;.
\end{align}
We have for any $t\ge0$
\begin{subequations}
\begin{align}
\hat{\rho}_{AM}(t) &= \left(\mathcal{N}_A\left(\frac{\lambda\, t}{\eta}\right)\otimes\mathbb{I}_M\right)(\hat{\rho}_{AM})\;,\\
\hat{\rho}_{BM}(t) &= \left(\mathcal{N}_B\left(\frac{(1-\lambda)\,t}{|1-\eta|}\right)\otimes\mathbb{I}_M\right)(\hat{\rho}_{BM})\;,\\
\hat{\rho}_{CM}(t) &= \left(\mathcal{N}_C(t)\otimes\mathbb{I}_M\right)(\hat{\rho}_{CM})\;,
\end{align}
\end{subequations}
where we have set $\hat{\rho}_{CM}:=\hat{\rho}_{CM}(0)$.
The time evolution preserves the condition $I(A:B|M)=0$.
Indeed, we have from the data-processing inequality for the quantum mutual information
\begin{equation}
0\le I(A:B|M)_{\hat{\rho}_{ABM}(t)} \le I(A:B|M)_{\hat{\rho}_{ABM}} = 0\;.
\end{equation}
We define the function
\begin{equation}
\phi(t) := S(C|M)_{\hat{\rho}_{CM}(t)} - \lambda\,S(A|M)_{\hat{\rho}_{AM}(t)} - (1-\lambda)\,S(B|M)_{\hat{\rho}_{BM}(t)}\;.
\end{equation}
We have from \autoref{prop:ddB} and \autoref{thm:Stam}
\begin{equation}
\phi'(t) = J(C|M)_{\hat{\rho}_{CM}(t)} - \frac{\lambda^2}{\eta}J(A|M)_{\hat{\rho}_{AM}(t)} - \frac{(1-\lambda)^2}{|1-\eta|}J(B|M)_{\hat{\rho}_{BM}(t)} \le 0\;.
\end{equation}
From \autoref{thm:conc}, $\phi$ is a linear combination of continuous concave functions, hence it is almost everywhere differentiable and for any $t\ge0$
\begin{equation}
\phi(t) - \phi(0) = \int_0^t\phi'(s)\,\mathrm{d}s \le 0\;.
\end{equation}
We then have from \autoref{thm:scaling}
\begin{align}
\phi(0) &\ge \lim_{t\to\infty}\phi(t)\nonumber\\
&= \lim_{t\to\infty}\left(S(C|M)_{\hat{\rho}_{CM}(t)} - \lambda\,S(A|M)_{\hat{\rho}_{AM}(t)} - (1-\lambda)\,S(B|M)_{\hat{\rho}_{BM}(t)}\right)\nonumber\\
&=n\left(\lambda\ln\frac{\eta}{\lambda} + \left(1-\lambda\right)\ln\frac{|1-\eta|}{1-\lambda}\right)\;,
\end{align}
and the quantum linear conditional Entropy Power Inequality \eqref{eq:EPIl} follows.
The quantum conditional Entropy Power Inequality \eqref{eq:EPI} follows choosing
\begin{equation}
\lambda = \frac{\eta\,\exp\frac{S(A|M)_{\hat{\rho}_{AM}(t)}}{n}}{\eta\,\exp\frac{S(A|M)_{\hat{\rho}_{AM}(t)}}{n} + \left|1-\eta\right|\exp\frac{S(B|M)_{\hat{\rho}_{BM}(t)}}{n}}\;,
\end{equation}
that maximizes the right-hand side of \eqref{eq:EPIl}.
\end{proof}

\section{Optimality of the quantum conditional Entropy Power Inequality}\label{sec:sharp}
In this Section, we prove that the quantum conditional Entropy Power Inequality is asymptotically saturated by a suitable sequence of quantum Gaussian input states.
\begin{theorem}[optimality of the quantum conditional Entropy Power Inequality]\label{thm:sharp}
The quantum conditional Entropy Power Inequality \eqref{eq:EPI} is optimal.
In other words, for any $a,\,b\in\mathbb{R}$ there exists a sequence of quantum Gaussian input states $\left\{\hat{\gamma}_{ABA'B'}^{(n)}\right\}_{n\in\mathbb{N}}$ of the form
\begin{equation}
\hat{\gamma}_{ABA'B'}^{(n)}=\hat{\gamma}_{AA'}^{(n)}\otimes \hat{\gamma}_{BB'}^{(n)}\;,\qquad n\in\mathbb{N}\;,
\end{equation}
such that
\begin{subequations}
\begin{align}
\lim_{n\to\infty}\exp\left(S(A|A'B')_{\hat{\gamma}_{AA'B'}^{(n)}}-1\right) &= a\;,\\
\lim_{n\to\infty}\exp\left(S(B|A'B')_{\hat{\gamma}_{BA'B'}^{(n)}}-1\right) &= b
\end{align}
\end{subequations}
and
\begin{equation}
\lim_{n\to\infty}\exp\left(S(C|A'B')_{\hat{\gamma}_{CA'B'}^{(n)}}-1\right) = \eta\,a + \left|1-\eta\right|b\;,
\end{equation}
where
\begin{equation}
\hat{\gamma}_{CA'B'}^{(n)} := \left(\mathcal{B}_\eta\otimes\mathbb{I}_{A'B'}\right)\left(\hat{\gamma}_{AA'}^{(n)}\otimes \hat{\gamma}_{BB'}^{(n)}\right)\;,
\end{equation}
and $A$, $A'$, $B$, $B'$ and $C$ are one-mode Gaussian quantum systems.
\end{theorem}
\begin{proof}
Let $\hat{\gamma}_{AA'}^{(n)}$ and $\hat{\gamma}_{BB'}^{(n)}$ be the quantum Gaussian states with covariance matrices
\begin{subequations}
\begin{align}
\sigma_{AA'}^{(n)} &= n\left(
                            \begin{array}{cc|cc}
                              \frac{n}{a} & 0 & \sqrt{\frac{n^2}{a^2}-1} & 0 \\
                              0 & \frac{n}{a} & 0 & -\sqrt{\frac{n^2}{a^2}-1} \\\hline
                              \sqrt{\frac{n^2}{a^2}-1} & 0 & \frac{n}{a} & 0 \\
                              0 & -\sqrt{\frac{n^2}{a^2}-1} & 0 & \frac{n}{a} \\
                            \end{array}
                          \right)\;,\\
\sigma_{BB'}^{(n)} &= n\left(
                            \begin{array}{cc|cc}
                              \frac{n}{b} & 0 & \sqrt{\frac{n^2}{b^2}-1} & 0 \\
                              0 & \frac{n}{b} & 0 & -\sqrt{\frac{n^2}{b^2}-1} \\\hline
                              \sqrt{\frac{n^2}{b^2}-1} & 0 & \frac{n}{b} & 0 \\
                              0 & -\sqrt{\frac{n^2}{b^2}-1} & 0 & \frac{n}{b} \\
                            \end{array}
                          \right)\;,
\end{align}
\end{subequations}
where the block decompositions refer to the bipartitions $AA'$ and $BB'$, respectively.
The symplectic eigenvalues of $\sigma_{AA'}^{(n)}$ are $(n,\,n)$, hence
\begin{subequations}
\begin{align}
S(AA')_{\hat{\gamma}_{AA'}^{(n)}} &= 2\,g\left(n-\frac{1}{2}\right) = \ln n^2 + 2 + \mathcal{O}\left(\frac{1}{n^2}\right)\;,\\
S(A')_{\hat{\gamma}_{AA'}^{(n)}} &= g\left(\frac{n^2}{a}-\frac{1}{2}\right) = \ln\frac{n^2}{a} + 1 + \mathcal{O}\left(\frac{1}{n^4}\right)\;,
\end{align}
\end{subequations}
and
\begin{equation}
\lim_{n\to\infty}S(A|A'B')_{\hat{\gamma}_{AA'B'}^{(n)}} = \lim_{n\to\infty}S(A|A')_{\hat{\gamma}_{AA'}^{(n)}} = 1 + \ln a\;.
\end{equation}
Analogously,
\begin{equation}
\lim_{n\to\infty}S(B|A'B')_{\hat{\gamma}_{BA'B'}^{(n)}} = 1 + \ln b\;.
\end{equation}
For $0\le\eta\le1$, the covariance matrix of $\hat{\gamma}^{(n)}_{CA'B'}$ is
\begin{equation}
\sigma^{(n)}_{CA'B'} = n\left(
                         \begin{array}{c|c|c}
                           n\left(\frac{\eta}{a}+\frac{1-\eta}{b}\right)I & \sqrt{\eta\left(\frac{n^2}{a^2}-1\right)}T & \sqrt{(1-\eta)\left(\frac{n^2}{b^2}-1\right)}T \\\hline
                           \sqrt{\eta\left(\frac{n^2}{a^2}-1\right)}T & \frac{n}{a}I & 0 \\\hline
                           \sqrt{(1-\eta)\left(\frac{n^2}{b^2}-1\right)}T & 0 & \frac{n}{b}I \\
                         \end{array}
                       \right),
\end{equation}
where
\begin{equation}
I = \left(
      \begin{array}{cc}
        1 & 0 \\
        0 & 1 \\
      \end{array}
    \right)\;,\qquad
T = \left(
      \begin{array}{cc}
        1 & 0 \\
        0 & -1 \\
      \end{array}
    \right)\;,
\end{equation}
and the block decomposition refers to the tripartition $CA'B'$.
We have on one hand
\begin{equation}
S(A'B')_{\hat{\gamma}^{(n)}_{CA'B'}} = g\left(\frac{n^2}{a}-\frac{1}{2}\right) + g\left(\frac{n^2}{b}-\frac{1}{2}\right) = \ln\frac{n^4}{a\,b} + 2 + \mathcal{O}\left(\frac{1}{n^4}\right)\;.
\end{equation}
Since $\frac{\sigma^{(n)}_{CA'B'}}{2n}$ is still the covariance matrix of a positive quantum Gaussian state, its symplectic eigenvalues are all larger than $\frac{1}{2}$, and the symplectic eigenvalues $\left(\nu_1^{(n)},\,\nu_2^{(n)},\,\nu_3^{(n)}\right)$ of $\sigma^{(n)}_{CA'B'}$ are all larger than $n$.
We then have
\begin{align}
S(CA'B')_{\hat{\gamma}^{(n)}_{CA'B'}} &= g\left(\nu_1^{(n)}-\frac{1}{2}\right) + g\left(\nu_2^{(n)}-\frac{1}{2}\right) + g\left(\nu_3^{(n)}-\frac{1}{2}\right)\nonumber\\
&= \ln\left(\nu_1^{(n)}\nu_2^{(n)}\nu_3^{(n)}\right) + 3 + \mathcal{O}\left(\frac{1}{n^2}\right)\nonumber\\
&= \frac{1}{2}\ln\det\sigma^{(n)}_{CA'B'} + 3 + \mathcal{O}\left(\frac{1}{n^2}\right)\nonumber\\
&= \ln\frac{n^4\left(\eta\,a+\left(1-\eta\right)b\right)}{a\,b} + 3 + \mathcal{O}\left(\frac{1}{n^2}\right)\;,
\end{align}
and
\begin{equation}
\lim_{n\to\infty}S(C|A'B')_{\hat{\gamma}^{(n)}_{CA'B'}} = \ln\left(\eta\,a+\left(1-\eta\right)b\right) + 1\;.
\end{equation}
Similarly, for $\eta\ge1$ the covariance matrix of $\hat{\gamma}^{(n)}_{CA'B'}$ is
\begin{equation}
\sigma^{(n)}_{CA'B'} = n\left(
                         \begin{array}{c|c|c}
                           n\left(\frac{\eta}{a}+\frac{1-\eta}{b}\right)I & \sqrt{\eta\left(\frac{n^2}{a^2}-1\right)}T & \sqrt{(\eta-1)\left(\frac{n^2}{b^2}-1\right)}I \\\hline
                           \sqrt{\eta\left(\frac{n^2}{a^2}-1\right)}T & \frac{n}{a}I & 0 \\\hline
                           \sqrt{(\eta-1)\left(\frac{n^2}{b^2}-1\right)}I & 0 & \frac{n}{b}I \\
                         \end{array}
                       \right).
\end{equation}
If $\left(\nu_1^{(n)},\,\nu_2^{(n)},\,\nu_3^{(n)}\right)$ are its symplectic eigenvalues,
\begin{align}
S(CA'B')_{\hat{\gamma}^{(n)}_{CA'B'}} &= g\left(\nu_1^{(n)}-\frac{1}{2}\right) + g\left(\nu_2^{(n)}-\frac{1}{2}\right) + g\left(\nu_3^{(n)}-\frac{1}{2}\right)\nonumber\\
&= \ln\left(\nu_1^{(n)}\nu_2^{(n)}\nu_3^{(n)}\right) + 3 + \mathcal{O}\left(\frac{1}{n^2}\right)\nonumber\\
&= \frac{1}{2}\ln\det\sigma^{(n)}_{CA'B'} + 3 + \mathcal{O}\left(\frac{1}{n^2}\right)\nonumber\\
&= \ln\frac{n^4\left(\eta\,a+\left(\eta-1\right)b\right)}{a\,b} + 3 + \mathcal{O}\left(\frac{1}{n^2}\right)\;,
\end{align}
and
\begin{equation}
\lim_{n\to\infty}S(C|A'B')_{\hat{\gamma}^{(n)}_{CA'B'}} = \ln\left(\eta\,a+\left(\eta-1\right)b\right) + 1\;.
\end{equation}
\end{proof}

\section{Entanglement-assisted classical capacity}\label{sec:capacity}
In this section, we exploit the quantum conditional Entropy Power Inequality to prove an upper bound to the entanglement-assisted classical capacity of the following non-Gaussian quantum channel.
This implication has been first considered in \cite{koenig2015conditional}, section III.

Let us fix a quantum state $\hat{\sigma}_B$ on the $n$-mode Gaussian quantum system $B$.
We consider the channel $\Phi:A\to C$ that mixes with $\hat{\sigma}_B$ the input state $\hat{\rho}_A$ on the $n$-mode Gaussian quantum system $A$ through a beam-splitter or a squeezing operation:
\begin{equation}
\Phi(\hat{\rho}_A) = \mathcal{B}_\eta(\hat{\rho}_A\otimes\hat{\sigma}_B)\;,\qquad \eta\ge0\;.
\end{equation}

If the sender can use an unlimited amount of energy, the entanglement-assisted classical capacity is infinite.
Since this scenario is not physical, we assume that the sender can use at most an energy $E$ per each mode.
The entanglement-assisted classical capacity \cite{holevo2013quantum,wilde2017quantum} of $\Phi$ is then equal to the supremum of the quantum mutual information:
\begin{equation}
C_{ea}(\Phi) = \sup\left\{I(C:M)_{(\Phi\otimes\mathbb{I}_M)(\hat{\rho}_{AM})}:\hat{\rho}_{AM}\text{ pure},\;\mathrm{Tr}_A\left[\hat{H}_A\,\hat{\rho}_A\right]\le n\,E\right\}\;.
\end{equation}
Let
\begin{equation}
E_0 := \frac{1}{n}\mathrm{Tr}_B\left[\hat{H}_B\,\hat{\sigma}_B\right]\;,\qquad S_0 := \frac{S(\hat{\sigma}_B)}{n}
\end{equation}
be the average energy and the entropy per mode of $\hat{\sigma}_B$, respectively.
The average energy per mode of $\Phi(\hat{\rho}_A)$ is
\begin{align}
\frac{1}{n}\mathrm{Tr}_C\left[\hat{H}_C\,\Phi(\hat{\rho}_A)\right] &= \frac{\eta}{n}\mathrm{Tr}_A\left[\hat{H}_A\,\hat{\rho}_A\right] + |1-\eta|\,E_0 + \frac{\eta+|1-\eta|-1}{2}\nonumber\\
&\le \eta\,E + |1-\eta|\,E_0 + \frac{\eta+|1-\eta|-1}{2}\;.
\end{align}
From \autoref{lem:Topt},
\begin{equation}
S(\Phi(\hat{\rho}_A)) \le n\;g\left(\eta\,E + |1-\eta|\,E_0 + \frac{\eta+|1-\eta|-1}{2}\right)\;.
\end{equation}
From the quantum conditional Entropy Power Inequality we have (we recall that $M$ is correlated only with $A$ and that $\hat{\rho}_{AM}$ is pure)
\begin{align}
\exp\frac{S(C|M)_{(\Phi\otimes\mathbb{I}_M)(\hat{\rho}_{AM})}}{n} &\ge \eta\exp\frac{S(A|M)_{\hat{\rho}_{AM}}}{n} + \left|1-\eta\right| \exp S_0\nonumber\\
&= \eta\exp\frac{-S(\hat{\rho}_A)}{n} + \left|1-\eta\right| \exp S_0\nonumber\\
&\ge \eta\exp(-g(E))+ \left|1-\eta\right| \exp S_0\;,
\end{align}
where in the last step we have used \autoref{lem:Topt} again.
Finally,
\begin{align}
I(C:M)_{(\Phi\otimes\mathbb{I}_M)(\hat{\rho}_{AM})} &= S(\Phi(\hat{\rho}_A)) - S(C|M)_{(\Phi\otimes\mathbb{I}_M)(\hat{\rho}_{AM})}\nonumber\\
&\le n\;g\left(\eta\,E + |1-\eta|\,E_0 + \frac{\eta+|1-\eta|-1}{2}\right)\nonumber\\
&\phantom{\le} - n\ln\left(\eta\,\mathrm{e}^{-g(E)}+ \left|1-\eta\right| \mathrm{e}^{S_0}\right)\;,
\end{align}
so that
\begin{align}
C_{ea}(\Phi) &\le n\;g\left(\eta\,E + |1-\eta|\,E_0 + \frac{\eta+|1-\eta|-1}{2}\right)\nonumber\\
&\phantom{\le} - n\ln\left(\eta\,\mathrm{e}^{-g(E)}+ \left|1-\eta\right| \mathrm{e}^{S_0}\right)\;.
\end{align}

\section{Conclusions}\label{sec:concl}
We have proven the conditional Entropy Power Inequality for Gaussian quantum systems, which are the most promising platform for quantum communication and quantum key distribution.
This fundamental inequality determines the minimum quantum conditional entropy of the output of the beam-splitter or of the squeezing among all the quantum input states where the two inputs are conditionally independent given the memory and have given quantum conditional entropies.
This inequality is optimal, since it is asymptotically saturated by a suitable sequence of quantum Gaussian input states.
In the unconditioned case, the optimal inequality is still an open challenging conjecture \cite{guha2008entropy} that is turning out to be very hard to prove \cite{de2015passive,de2016passive,de2016pq,de2016gaussiannew,de2016gaussian,de2017multimode}.
The quantum conditional Entropy Power Inequality instead definitively settles the problem in the conditioned case.

\begin{acknowledgements}
GdP acknowledges financial support from the European Research Council (ERC Grant Agreements no 337603 and 321029), the Danish Council for Independent Research (Sapere Aude) and VILLUM FONDEN via the QMATH Centre of Excellence (Grant No. 10059).
\end{acknowledgements}

\appendix
\section{Proof of the classical conditional Entropy Power Inequality}\label{app:cCEPI}
We have from the definition of the classical conditional entropy
\begin{equation}\label{eq:C|M}
S(C|M) = \int_M S(C|M=m)\,\mathrm{d}p_M(m)\;,
\end{equation}
where $p_M$ is the probability distribution of $M$.
From the classical Entropy Power Inequality \eqref{eq:cEPI} we have for any $m$
\begin{equation}\label{eq:EPIMm}
S(C|M=m) \ge \frac{k}{2}\ln\left(\eta\exp\frac{2S(A|M=m)}{k} + \left|1-\eta\right|\exp\frac{2S(B|M=m)}{k}\right)\;.
\end{equation}
Finally, since the function
\begin{equation}
\left(a,\,b\right)\mapsto \frac{k}{2}\ln\left(\eta\exp\frac{2a}{k} + \left|1-\eta\right|\exp\frac{2b}{k}\right)\;,\qquad a,\,b\in\mathbb{R}
\end{equation}
is convex, we have from \eqref{eq:C|M}, \eqref{eq:EPIMm} and Jensen's inequality
\begin{align}\label{eq:cCEPIproof}
S(C|M) &\ge \int_M\frac{k}{2}\ln\left(\eta\,\mathrm{e}^{\frac{2}{k}S(A|M=m)} + \left|1-\eta\right|\mathrm{e}^{\frac{2}{k}S(B|M=m)}\right)\mathrm{d}p_M(m)\nonumber\\
& \ge \frac{k}{2}\ln\left(\eta\,\mathrm{e}^{\frac{2}{k}\int_M S(A|M=m)\mathrm{d}p_M(m)} + \left|1-\eta\right|\mathrm{e}^{\frac{2}{k}\int_M S(B|M=m)\mathrm{d}p_M(m)}\right)\nonumber\\
& = \frac{k}{2}\ln\left(\eta\,\mathrm{e}^{\frac{2}{k}S(A|M)} + \left|1-\eta\right|\mathrm{e}^{\frac{2}{k}S(B|M)}\right)\;.
\end{align}
The conditional Entropy Power Inequality \eqref{eq:cCEPI} is saturated iff all the inequalities in \eqref{eq:cCEPIproof} are equalities.
The first inequality is saturated iff, conditioning on any value $m$ of $M$, $A$ and $B$ are independent Gaussian random variables with proportional covariance matrices.
The second inequality is saturated iff $S(A|M=m)-S(B|M=m)$ does not depend on $m$.
For $A$ and $B$ as above, this is equivalent to having the proportionality constant between their covariance matrices independent on $m$.

\section{}\label{app}
\begin{lemma}[\cite{holevo2013quantum}, Lemma 12.25]\label{lem:GS}
Let $\hat{\rho}$ be a quantum state on a Gaussian quantum system with finite average energy, and let $\hat{\gamma}$ be the Gaussian quantum state with the same first and second moments.
Then
\begin{equation}
S(\hat{\gamma}) = S(\hat{\rho}) + S(\hat{\rho}\|\hat{\gamma})\;.
\end{equation}
\end{lemma}

\begin{lemma}\label{lem:Topt}
Let $\hat{\rho}$ be a quantum state on a Gaussian quantum system with finite average energy, and let $\hat{\omega}$ be the thermal Gaussian quantum state with the same average energy.
Then
\begin{equation}
S(\hat{\omega}) \ge S(\hat{\rho})\;.
\end{equation}
\end{lemma}
\begin{proof}
Let $\beta>0$ be such that
\begin{equation}
\hat{\omega} = \frac{\mathrm{e}^{-\beta\hat{H}}}{\mathrm{Tr}\,\mathrm{e}^{-\beta\hat{H}}}\;.
\end{equation}
We then have
\begin{equation}
S(\hat{\omega}) = S(\hat{\rho}) + S(\hat{\rho}\|\hat{\omega}) + \beta\,\mathrm{Tr}\left[\hat{H}\left(\hat{\omega}-\hat{\rho}\right)\right] \ge S(\hat{\rho})\;.
\end{equation}
\end{proof}

\begin{lemma}\label{lem:diff}
Let $A$ and $B$ be Gaussian quantum systems with $m$ and $n$ modes, respectively, and $\Phi:A\to B$ a Gaussian quantum channel.
Let $\hat{\rho}_A$ be a quantum state on $A$, and $p_X$ a probability measure on $\mathbb{R}^{2m}$.
We define
\begin{equation}
\hat{\sigma}_A := \int_{\mathbb{R}^{2m}}\hat{D}_A(\mathbf{x})\,\hat{\rho}_A\,{\hat{D}_A(\mathbf{x})}^\dag\,\mathrm{d}p_X(\mathbf{x})\;.
\end{equation}
Then,
\begin{equation}
S(\Phi(\hat{\rho}_A)) - S(\hat{\rho}_A) \ge S(\Phi(\hat{\sigma}_A)) - S(\hat{\sigma}_A)\;.
\end{equation}
\end{lemma}
\begin{proof}
Let $\hat{\sigma}_{AX}$ be the joint state on $AX$ such that its marginal on $X$ is $p_X$, and for any $\mathbf{x}\in\mathbb{R}^{2m}$
\begin{equation}
\hat{\sigma}_{A|X=\mathbf{x}} = \hat{D}_A(\mathbf{x})\,\hat{\rho}_A\,{\hat{D}_A(\mathbf{x})}^\dag\;.
\end{equation}
We notice that the marginal of $\hat{\sigma}_{AX}$ on $A$ is $\hat{\sigma}_A$, and
\begin{align}
S(A|X)_{\hat{\sigma}_{AX}} &= \int_{\mathbb{R}^{2m}}S\left(\hat{\sigma}_{A|X=\mathbf{x}}\right)\mathrm{d}p_X(\mathbf{x}) = S(\hat{\rho}_A)\;,\\
S(B|X)_{(\Phi\otimes\mathbb{I}_X)(\hat{\sigma}_{AX})} &= \int_{\mathbb{R}^{2m}}S\left(\Phi(\hat{\sigma}_{A|X=\mathbf{x}})\right)\mathrm{d}p_X(\mathbf{x}) = S(\Phi(\hat{\rho}_A))\;,
\end{align}
where in the last step we used that for any $\mathbf{x}\in\mathbb{R}^{2m}$
\begin{equation}
\Phi\left(\hat{D}_A(\mathbf{x})\,\hat{\rho}_A\,{\hat{D}_A(\mathbf{x})}^\dag\right) = \hat{D}_B(K\mathbf{x})\,\Phi(\hat{\rho}_A)\,{\hat{D}_B(K\mathbf{x})}^\dag
\end{equation}
for some matrix $K:\mathbb{R}^{2m}\to\mathbb{R}^{2n}$.
We then have
\begin{align}
S(\hat{\sigma}_A) - S(\hat{\rho}_A) &= I(A:X)_{\hat{\sigma}_{AX}}\;,\\
S(\Phi(\hat{\sigma}_A)) - S(\Phi(\hat{\rho}_A)) &= I(B:X)_{(\Phi\otimes\mathbb{I}_X)(\hat{\sigma}_{AX})}\;,
\end{align}
and the claim follows from the data-processing inequality for the mutual information.
\end{proof}

\bibliography{biblio}
\bibliographystyle{plain}

\end{document}